\documentclass[12pt,dvipsnames]{article}

\usepackage[margin=1in]{geometry}  
\usepackage{graphicx}              
\usepackage{amsmath}               
\usepackage{amsfonts}              
\usepackage{amsthm}                
\usepackage{breqn}
\usepackage{amsmath, latexsym, amsfonts, amssymb}
\usepackage{mathrsfs,enumerate}
\usepackage{graphics,epsfig,color}
\usepackage{xcolor}
\usepackage{subfigure}
\usepackage{wrapfig}
\usepackage{cite}
\usepackage{caption}
\usepackage{authblk}

\usepackage{calrsfs}
\usepackage{bm}

\def\non{\nonumber}

\newtheorem{theorem}{Theorem}[section]
\newtheorem{lemma}[theorem]{Lemma}
\newtheorem{prop}[theorem]{Proposition}

\title{Large time asymptotic of heavy tailed renewal processes}

\author[1]{Hiroshi Horii}
\author[1]{Rapha\"el Lefevere}
\author[2]{Takahiro Nemoto}

\affil[1]{Universit\'e de Paris, Laboratoire de Probabilit\'es, Statistiques et Mod\'elisation, UMR 8001, F-75205 Paris, France}
\affil[2]{Sorbonne Universit\'e, INSERM, CNRS, Institut de la Vision, 75012 Paris, France}

\date{}

\begin{document}

\maketitle

\begin{abstract}

We study the large-time asymptotic of renewal-reward processes with a heavy-tailed waiting time distribution. It is known that the heavy tail of the distribution produces an extremely slow dynamics, resulting in a singular large deviation function. When the singularity takes place, the bottom of the large deviation function is flattened, manifesting anomalous fluctuations of the renewal-reward processes. In this article, we aim to study how these singularities emerge as the time increases. Using a classical result on the sum of random variables with regularly varying tail, we develop an expansion approach to prove an upper bound of the finite-time moment generating function for the Pareto waiting time distribution (power law) with an integer exponent. We  perform numerical simulations using Pareto (with a real value exponent), inverse Rayleigh and log-normal waiting time distributions, and demonstrate  similar results are anticipated in these  waiting time distributions.

\end{abstract}

\section{Introduction}

Let us consider a random walk with arbitrary distributions of jump lengths and waiting times. 
Both waiting times $(\tau_i)$, $i=1,2,...$ and jump lengths $(X_i)$, $i=1,2,...$ are {\it i.i.d.} whose probability densities are denoted by respectively $p$ (such that $\tau>0$ a.s. and $\mathbb{E}\tau<\infty$) and $q$. 
The {\it renewal-reward process} $R(t)$ is then defined as
\begin{equation}
	R(t)= X_1+\cdots + X_n, 
\end{equation}
where the number of jumps $n$ satisfies
\begin{equation}
 \tau_1+\cdots+\tau_n\leq t<\tau_1+\cdots+\tau_{n+1}.
 \end{equation}
A simple example of the renewal-reward process is a counting process $N_t$ defined by $q(x)=\delta(x-1)$, (so that $X_i$ only takes the value 1), {\it i.e.} $N_t$ corresponds to the number of jumps until time $t$.

The well-known fields that exploit the renewal-reward process (and the theory related to this process) are, among others, the actuarial science, where models describing an insurer's vulnerability to ruin are studied (known as ruin theory) \cite{b22}, the queueing theory that studies the queue length and waiting time in telecommunication, traffic- and industrial engineering \cite{b23}, and epidemiology using a renewal-reward process for estimating the basic quantity of virus spreading \cite{b11,b12}.  A number of studies have also reported various physical and biological phenomena which are described by renewal-reward processes  \cite{klafter1980derivation, tanushev1997central, berkowitz2006modeling, b5}.

In this article, we study the fluctuations of the renewal-reward process $R(t)$.  
Under mild assumptions \cite{b19}, the family of random variables $(R(t)/t)_{t>0}$ satisfies a large deviation principle 
\begin{equation}
	\mathbb P \left (\frac {R(t)}{t}\simeq s \right ) \sim e^{- t I(s)}, \quad {\mathrm as}\quad t\to\infty,
\end{equation}
where the non zero function $I$ is called the rate function. See \cite{b19} for more mathematical formulation of the large deviation principle. In good cases \cite{b19}, the rate function $I$ is equal to the Legendre transform of the scaled cumulant generating function (CGF)
\begin{equation}
	\phi(h)=  \lim_{t\rightarrow \infty}\frac{1}{t} {\rm log}\mathbb{E}[e^{hR(t)}],
\end{equation}
{\it i.e.}, $I(s) = \sup \{ h s - \phi(h) \}$. Here $h\in\mathbb{R}$ is called a biasing field.  

An interesting point to keep in mind when studying renewal-reward processes is that the process is in general not Markov and  a strong time correlation in the dynamics can be introduced by choosing a heavy-tailed probability density as the waiting time distribution. One can intuitively expect that a single waiting time could become dominant in the dynamics because of the heavy tail in the waiting time distribution, leading to an unusually high occurrence of certain rare events (see, for example, a single-big jump principle \cite{bouchaud1990anomalous, kutner2002extreme, de2013asymmetric, b6}).
Considering renewal-reward processes with heavy-tailed waiting-time distributions, the presence of a singular behaviour in the rate function and the CGF was established in \cite{b19,b20,b14} related to these rare events. For example, let us consider the counting process $N_t$ with waiting times distributed according to a Pareto's law with a parameter $m>2$:
\begin{eqnarray}
	p(t):&=&\left\{\begin{array}{ll}0 & t\leq 0\\ 
	\frac{(m-1)}{(1+t)^m} & t > 0
	\end{array}
\right.
\label{eq:pareto}
\end{eqnarray}
or equivalently by its cumulative distribution $F$
\begin{eqnarray}
\non\\
	F(t):=\mathbb P[\tau\leq t]&=&\left\{\begin{array}{ll}0 & t\leq 0\\ 
	1-\frac{1}{(1+t)^{m-1}}& t >0.
\end{array}
\right.
\end{eqnarray}
Then the result of \cite{b19} implies that
\begin{eqnarray}
\label{eq:affinepart}
&&I(x)=0\ \ \ \ \ \ \ (0\leq x\leq \mu)\non\\
 &&\phi(h)=0\ \ \ \ \ \ \ (h\leq 0),
\end{eqnarray}
where $\mu=1/\mathbb E[\tau]$. 
The fact that the rate function $I$ takes the value zero below $\mu=1/\mathbb E[\tau]$, indicates an unusually high occurrence of the rare events where $N_t$ is below its average value $\mu=1/\mathbb E[\tau]$. Throughout this article, we call this singularity an {\it affine part}.

Our goal in this article is to study the finite time asymptotic of CGF when the affine part emerges.  
The finite time asymptotics of the CGF and of the rate function in the context of renewal-reward processes have been studied by Tsirelson \cite{b2} under the assumption that the rewards are centred, namely $\mathbb E[X_i]=0, i\in\mathbb N$.  Tsirelson theorem however cannot be applied to the counting process $N_t$ as it satisfies $X_i=1, i\in\mathbb N$. The affine part is indeed not observed in the Tsirelson theorem \cite{b2}. To approach this problem, we rely on a classical result on the sum of random variables with regularly varying tails that can be found in Feller's book \cite{feller}.
The behaviour for $t\to\infty$ of $\mathbb P[N_t\leq k]$ is determined by 
 \begin{equation}
\lim_{t\to\infty}t^{m-1}\mathbb P[N_t \leq k-1]=k.
\end{equation}
for any $k\geq 1$. 
Using this theorem, we study the asymptotics in time of the cumulant generating function.
The strategy is that we first expand the moment generating function of $N_t$ as the infinite sum of $P[N_t <k]$ ($k=1,2,...$) and then use this theorem in each term. Technically, the key in this step is to justify the fact that the infinite sum and the large $t$ limit  can be exchanged. This leads to the asymptotic form of the upper bound for the finite-time moments generating function (MGF) (Theorem \ref{thm:normconv_0} and Theorem \ref{thm:normconv}) for Pareto waiting time distribution with an integer exponent. In numerical simulations we then test the validity of this theorem using Pareto distribution with a real value exponent, inverse Rayleigh distribution and log-normal distribution, demonstrating the extent of this theorem beyond what we study mathematically in this article.

The structure of this article is the following:  In Section \ref{seq:results}, we will state our results on the speed of convergence of the moments generating function (MGF).  
In Section \ref{seq:unibound}, we will show how to prove this theorem and derive the asymptotic form of the MGF with Pareto distribution.
In Section \ref{seq:appendixa}, we will perform numerical simulations and study if the theorem in Section \ref{seq:results} can be extended to the problems with different heavy-tailed waiting-time distributions. 
In Section \ref{sec:numerical_ratefunction}, we numerically study the rate function of $N_t$, observe a general asymptotic form, and discuss the relation with the results in Section \ref{seq:results}.
Finally in Section \ref{seq:conclusion}, we will conclude this article.

\section{Results}
\label{seq:results}

Our goal is to study for $h<0$ :
\begin{equation}
	M(t,h):=\mathbb E[e^{hN_t}]=\sum_{k=0}^\infty e^{hk}\mathbb P[N_t=k].
\label{defM}
\end{equation}
We observe that, since obviously $0\leq \mathbb P[N_t=k]\leq 1$, the series is normally convergent for any $h<0$, as a function defined for $t\in[0,\infty[$. Let 
$$
S_k=\tau_1+\ldots+\tau_k
$$
Since for any $t\geq 0$, any $k\in\mathbb N\backslash\{0\}$, 
$$
\mathbb P[N_t=k]=\mathbb P[N_t\leq k]-\mathbb P[N_t\leq k-1]
$$ 
and any $k\in\mathbb N$ :
$$
\mathbb P[N_t\leq k]=\mathbb P[S_{k+1}\geq t],
$$
we obtain :
$$
M(t,h)=\sum_{k=1}^\infty (e^{h(k-1)}-e^{hk})\mathbb P[S_k\geq t].
$$
Introducing $z=e^h$, this may be rewritten :

$$
M(t,h)=\frac{1-z}{z}\sum_{k=1}^\infty z^k \mathbb P[S_k\geq t]
$$
The behaviour for $t\to\infty$ of $P[S_k\geq t]$ when $S_k$ is the sum of independent variables distributed with a regularly varying distribution may be found in Feller's book \cite{feller} Chapter VIII.9 p. 278.  The Pareto distribution with $m\geq 3$ enters this class and Feller's result implies that for any $k\geq 1$ :
\begin{equation}
	\lim_{t\to\infty}t^{m-1}\mathbb P[S_k>t]=k.
\label{feller1}
\end{equation}

Thus, if one is allowed to take the limit inside the above series, one gets the following

\begin{theorem}\label{thm:normconv_0} 
Let $(N_t)_{t\geq0}$ the counting process with waiting times distributed according to a Pareto's law with an integer parameter $m\geq 3$, then for any $h<0$ :
$$
\lim_{t\to\infty}t^{m-1} M(t,h)=\frac 1{1-e^h}.
$$

\end{theorem}
The proof of the theorem boils down to proving that this exchange is justified and using $\sum_{k=1}^\infty kz^k=\frac z {(1-z)^2}$ for $0\leq z <1$.  The justification of the exchange of the series and the limit is guaranteed by (\ref{normconv}) in the next theorem.

\begin{theorem}\label{thm:normconv} 
Let $(N_t)_{t\geq0}$ the counting process with waiting times distributed according to a Pareto's law with an integer parameter $m\geq 3$, then for any $h<0$
\begin{equation}
	\sum_{k=1}^\infty e^{hk} \sup_{t\in [0,\infty[}t^{m-1}\mathbb P[S_k\geq t]<\infty.
\label{normconv}
\end{equation}
Moreover, for any $m\geq 3$, there exists $\bar c$ such that for any $h<0$, and for any $d>\frac{\bar c}{e^{-h}-1}$ :
\begin{equation}
	M(h,t)\leq \frac{1}{1-e^h}M_0(t)+C_m\frac{d^m}{(d+t)^m}\frac{\alpha(d) e^h}{(1-\alpha(d) e^h)^2}
	\label{unifbound}
\end{equation}
where $\alpha(d)=1+\frac{\bar c}{d}$.
\end{theorem}

\noindent {\bf Remark}
Note that the first term in (\ref{unifbound}) is $O(\frac 1{t^{m-1}})$ while the second term is $O(\frac 1{t^{m}})$. 
\begin{proof}
We introduce the notation $M_k(t)=\mathbb P[N_t=k]$ for $k\in\mathbb N$.
We note the two straightforward relations : for any $k\geq 1$
\begin{eqnarray}
	M_0(t)&=&1-F(t)\nonumber\\
	M_k(t) &=& \int_0^t M_{k-1}(t-s)p(s) ds
	\label{recursion1}
\end{eqnarray}
	We use the telescopic identity that holds for any $n\geq 1$
	\begin{eqnarray}
\label{eq:identity}
M_n(t)=M_0(t)+\sum^n_{k=1}(M_k(t)-M_{k-1}(t)).
\end{eqnarray}
which, together with Proposition \ref{thm:Msub1} below implies that for any integer $m\geq 3$ there exists $\bar c>0$ such that for any $d>1$, any $t\geq 0$ and for $n\geq 1$:
\begin{equation}
\mathbb P[N_t=n]= M_n(t)\leq M_0(t)+n C_m\frac{d^m}{(d+t)^m}\left[1+\frac{\bar c}{d}\right]^{n-1}
\label{nt}	
\end{equation}
and  
\begin{eqnarray}
	\mathbb P[S_n\geq t]&=&\mathbb P[N_t\leq n-1]\nonumber\\
	&=&\sum_{k=0}^{n-1}\mathbb P[N_t=k]\nonumber\\
	&\leq & n M_0(t)+n^2 C_m\frac{d^m}{(d+t)^m}\left[1+\frac{\bar c}{d}\right]^{n-1}.
\label{Snbound}
\end{eqnarray}
Thus for any $h<0$, and any $d>\frac{\bar c}{e^{-h}-1}$ the series (\ref{normconv}) converges and the bound (\ref{unifbound})holds.
\end{proof}

\noindent {\bf Remarks}

\noindent If $m\geq 3$ is not an integer, it is easy to see that
$$
\mathbb P[S_k\geq t]\leq \mathbb P[\tilde S_k\geq t]
$$ 
where $\tilde S_k=\tilde\tau_1+\cdots+\tilde\tau_k$ and $(\tau_k)_k$ is and i.i.d sequence such that
$$
{\mathbb P}[\tilde\tau_i>t]=\frac{1}{(1+t)^{\lfloor m\rfloor -1}}
$$
because we have that
$$
{\mathbb P}[\tau_i>t]\leq {\mathbb P}[\tilde\tau_i>t].
$$
Thus, using the same methods, we can conclude that 
$$
\lim_{t\to\infty} t^{\lfloor m\rfloor-1}M(t,h)=0
$$
but not more than that.  We expect however the theorem to be true also in the non-integer case. (See Section~\ref{seq:appendixa} for the study based on numerical simulations.)

\section{Uniform bounds}
\label{seq:unibound}

\begin{prop}\label{thm:Msub1} For any integer $m\geq 3$, there exists $\bar c>0$ such that for any $n\geq 1$, $d\geq 1$, and any $t>0$ :
\begin{eqnarray}
M_n(t)-M_{n-1}(t)\leq C(m)\frac{d^m}{(d+t)^m}\left[1+\frac{\bar c}{d}\right]^{n-1}.
\end{eqnarray}
\end{prop}
\begin{proof}
	We first prove for the case $n=1$ and then proceed by induction. First, $M_0(t)$ and $M_1(t)$ are given by
\begin{eqnarray}
M_0(t)&=&1-F(t)=\frac{1}{(1+t)^{m-1}}\non\\
M_1(t)&=&\int_0^{t}M_0(s)p(t-s)ds\non\\
&=&\int_{0}^{t}\frac{1}{(1+s)^{m-1}}\frac{(m-1)ds}{(1+t-s)^{m}}.\label{M1},
\end{eqnarray}
In order to calculate this integral, we perform a partial-fraction decomposition by viewing the denominator in the integral of (\ref{M1}) as a polynomial in $s$:
\begin{dmath}
	\label{eq:decom}
	\frac{1}{(1+s)^{m-1}(1+t-s)^{m}}=\sum_{k=1}^{m-1}\frac{a_k(t)}{(1+s)^k}+\sum_{k=1}^{m}\frac{b_k(t)}{(1+t-s)^k}
	\end{dmath}
	where 
	\begin{eqnarray}
		a_k(t)&=&\frac{1}{(m-1-k)!}\lim_{s\to -1}\frac{d^{m-1-k}}{ds^{m-1-k}}\left(\frac{1}{(1+t-s)^{m}}\right),\quad 1\leq k\leq m-1 \non\\
		b_k(t)&=& \frac{1}{(m-k)!}\lim_{s\to 1+t}\frac{d^{m-k}}{ds^{m-k}}\left(\frac{1}{(1+s)^{m-1}}\right),\quad 1\leq k\leq m.\label{akbk}
	\end{eqnarray}
And thus :
\begin{eqnarray*}
	a_k(t)&=&\frac{A_k(m)}{(2+t)^{2m-1-k}}\\
	b_k(t)&=&\frac{B_k(m)}{(2+t)^{2m-k}}.
\end{eqnarray*}
where the $A_k(m)$ and $B_k(m)$ are combinatorial factors.
Now with $A'_k(m)=(m-1) A_k(m)$ and $B'_k(m)=(m-1) B_k(m)$ we have 
\begin{equation}
	\label{expM1}
	M_1(t)=\sum_{k=1}^{m-1} \frac{A'_k(m)}{(2+t)^{2m-k}}\int_0^t \frac{ds}{(1+s)^k}+\sum_{k=1}^m \frac{B'_k(m)}{(2+t)^{2m-k-1}}\int_0^t \frac{ds}{(1+t-s)^k}
\end{equation}
 From (\ref{akbk}), we see that $B_k(m)=1$ and thus $B'_k(m)=(m-1)$.  Therefore in the second sum the term corresponding to $k=m$ is
 
	\begin{eqnarray}
	\frac{m-1}{(2+t)^{m-1}}\int_0^t \frac{ds}{(1+t-s)^m}&=&\frac{1}{(2+t)^{m-1}}(1-\frac{1}{(1+t)^{m-1}})\non\\
	&\leq & \frac{1}{(1+t)^{m-1}}=M_0(t).
\end{eqnarray}

It is easy to see that for $k>1$, the remaining integrals are all bounded by $1$. 
Thus we have :
\begin{eqnarray}
	M_1(t)\leq M_0(t)&+&\sum_{k=2}^{m-1} \frac{A'_k(m)}{(2+t)^{2m-k}}+
	\sum_{k=2}^{m-1} \frac{B'_k(m)}{(2+t)^{2m-k-1}}\non\\
	&+& \frac{A'_1(m)}{(2+t)^{2m-1}}\log(1+t)+\frac{B'_1(m)}{(2+t)^{2m-2}}\log(1+t)
	\end{eqnarray}
and because $\log(1+t)\leq (2+t)$ and $m\geq 3$ :
	
\begin{eqnarray}
	M_1(t)&\leq& M_0(t)+\frac{C(m)}{(2+t)^m}\non\\
	&\leq & M_0(t)+\frac{C(m)}{(1+t)^m}	\non\\
	&\leq & M_0(t)+C(m)\frac{d^m}{(d+t)^m}
\end{eqnarray}
for some constant $C(m)$ depending on $m$ and for any $d\geq 1$.	

Thus, we find the bound we were looking for :
\begin{equation}
M_1(t)-M_0(t) \leq C(m)\frac{d^m}{(d+t)^m}.
\label{eq:gmbound}
\end{equation}
We now proceed to prove the claim of the proposition for general $n>1$ by induction :
we see that
\begin{equation}
	M_{n}(t)-M_{n-1}(t)=\int_0^t (M_{n-1}(s)-M_{n-2}(s))p(t-s)
	\label{recursion}
\end{equation}
 and we assume that for some $\bar c>0$ :
$$
M_{n-1}(t)-M_{n-2}(t)\leq C(m)\frac{d^m}{(d+t)^m}\left[1+\frac{\bar c}{d}\right]^{n-2}
$$
for any $d\geq 1$ and any $t\geq 0$.
Then we conclude the proof of the proposition with the use of the following Lemma.

\begin{lemma} \label{thm:mnminusmn1} There exists $\bar c>0$, such that for any $d\geq 1$ and any $t\geq 0$:
	\begin{equation}
	\int^t_0 \frac{1}{(d+s)^m}\frac{(m-1)ds}{(1+t-s)^m}\leq \left[1+\frac{\bar c}{d}\right]\frac{1}{(d+t)^m}
	\end{equation}
	\end{lemma}

\end{proof}

We give now the proof of Lemma \ref{thm:mnminusmn1}.
\begin{proof} For $d\geq 1$, $t\geq 0$, let
\begin{equation}
	I_m(d,t)=\int^t_0 \frac{1}{(d+s)^m}\frac{(m-1)ds}{(1+t-s)^m}.
	\end{equation}

We first perform partial fraction decomposition in the integrand of  $I_m(d,t)$,
\begin{dmath}
	\frac{1}{(d+s)^m}\frac{m-1}{(1+t-s)^m}=\frac{(m-1)}{(1+d+t)^m}\left[\left(\frac{1}{(d+s)^m}+\frac{1}{(1+t-s)^m}\right)+\frac{L_1}{(1+d+t)}\left(\frac{1}{(d+s)^{m-1}}+\frac{1}{(1+t-s)^{m-1}}\right)+\cdots+\frac{L_{m-1}}{(1+d+t)^{m-1}}\left(\frac{1}{(d+s)}+\frac{1}{(1+t-s)}\right)\right],
\end{dmath}
where $L_i\ (i=1,2,\ldots,m-1)$ are constants. 
Integrating these terms over the interval between $s=0$ and $s=t$, we then get
\begin{dmath}
	I_m(d,t)=\frac{(m-1)}{(1+d+t)^m}\left[\frac{1}{m-1}\left(\frac{1}{d^{m-1}}+1-\frac{2}{(d+t)^{m-1}}\right)+\frac{L_1}{(1+d+t)}\frac{1}{m-2}\left(\frac{1}{d^{m-2}}+1-\frac{2}{(d+t)^{m-2}}\right)+\cdots+\frac{L_{m-1}}{(1+d+t)^{m-1}}\left({\rm log}\left(\frac{(d+t)(1+t)}{d}\right)\right)\right]\\
	\leq\frac{1}{(d+t)^m}\left[1+\frac{\bar c}{d}\right],
\end{dmath}
for some $\bar c>0$. We have used the relation $t>{\rm log}(t)$  to derive the inequality in the second line.

\end{proof}


\section{Numerical study}

\subsection{Generality of Theorem \ref{thm:normconv_0}}
\label{seq:appendixa}

In this section, we numerically study the validity of Theorem \ref{thm:normconv_0} beyond its hypotheses. 
We first test if our numerical simulations capture correctly Theorem \ref{thm:normconv_0} by using Pareto distribution with an integer value exponent (with which Theorem \ref{thm:normconv_0} was proven). We then numerically study the validity of this theorem for different waiting time distributions, such as Pareto distribution with a real value exponent, inverse Rayleigh distribution and lognormal distribution.

\begin{figure}
	\begin{center}
		\includegraphics[clip,width=8cm]{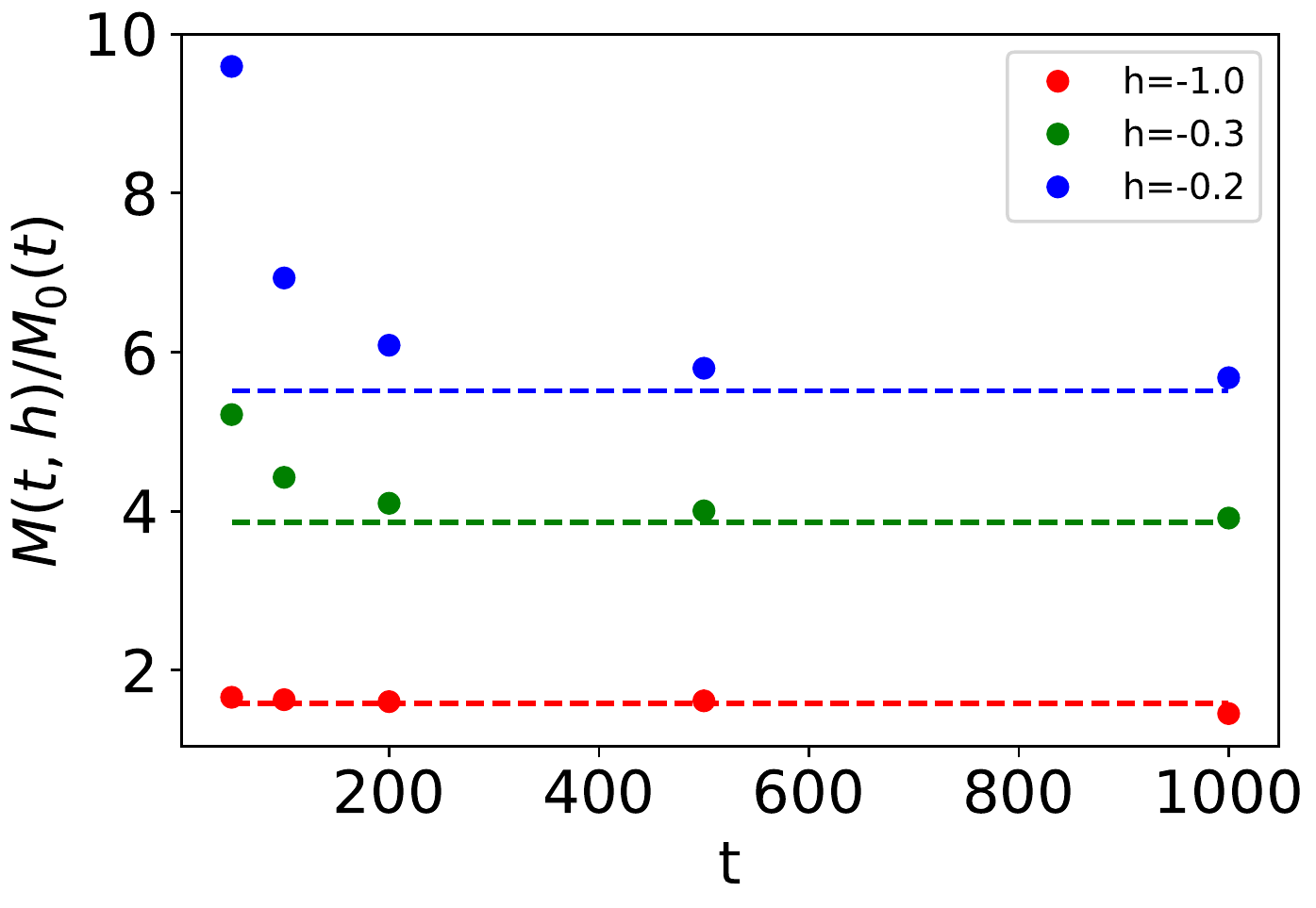}  
		\caption{ The numerical results of $M(t,h)/M_0(t)$ with the Pareto distribution ($m=3$).  Dashed lines are $1/(1-e^h)$. 
		}
		\label{fig:d}
	\end{center}
\end{figure}

Using numerical simulations, we estimate $M(t,h)$ for the Pareto distribution (with $m=3$) and divide it by $M_0(t)=1/(1+t)^{m-1}$. We plot $M(t,h)/M_0(t)$ as a function of $t$ in Fig.\ref{fig:d} together with $1/(1-e^h)$ as horizontal dashed lines. This demonstrates the reliability of our numerical simulations.

Next we consider the Pareto distribution with a real value exponent ($m=3.5$),  inverse Rayleigh distribution
\begin{eqnarray}
	p_{Ray}(t):&=&\left\{\begin{array}{ll}0 & t\leq 0\\ 
		\frac{\beta}{t^3}e^{-\frac{\beta}{2t^2}} & t > 0
	\end{array}
	\right.
	\label{eq:Rayleigh}
\end{eqnarray}
with a parameter $\beta$ and log-normal distribution
\begin{eqnarray}
	p_{log}(t):&=&\left\{\begin{array}{ll}0 & t\leq 0\\ 
		\frac{1}{\sqrt{2\pi}\sigma t}e^{-\frac{({\rm log}(t)-\mu)^2}{2\sigma^2}} & t > 0
	\end{array}
	\right.
	\label{eq:log-normal}
\end{eqnarray}
with parameters $\mu$ and $\sigma$. The corresponding cumulative distributions for the latter two distributions are  
\begin{eqnarray}
	\non\\
	F_{Ray}(t):=\mathbb P[\tau\leq t]&=&\left\{\begin{array}{ll}0 & t\leq 0\\ 
		e^{-\frac{\beta}{2t^2}}& t >0.
	\end{array}
	\right.
\end{eqnarray}
and
\begin{eqnarray}
	\non\\
	F_{log}(t):=\mathbb P[\tau\leq t]&=&\left\{\begin{array}{ll}0 & t\leq 0\\ 
		\frac{1}{2}{\rm erfc}\left[-\frac{{\rm log}(t)-\mu}{\sqrt{2}\sigma}\right]& t >0,
	\end{array}
	\right.
\end{eqnarray}
respectively.

\begin{figure}
	\begin{center}
		\includegraphics[clip,width=7.5cm]{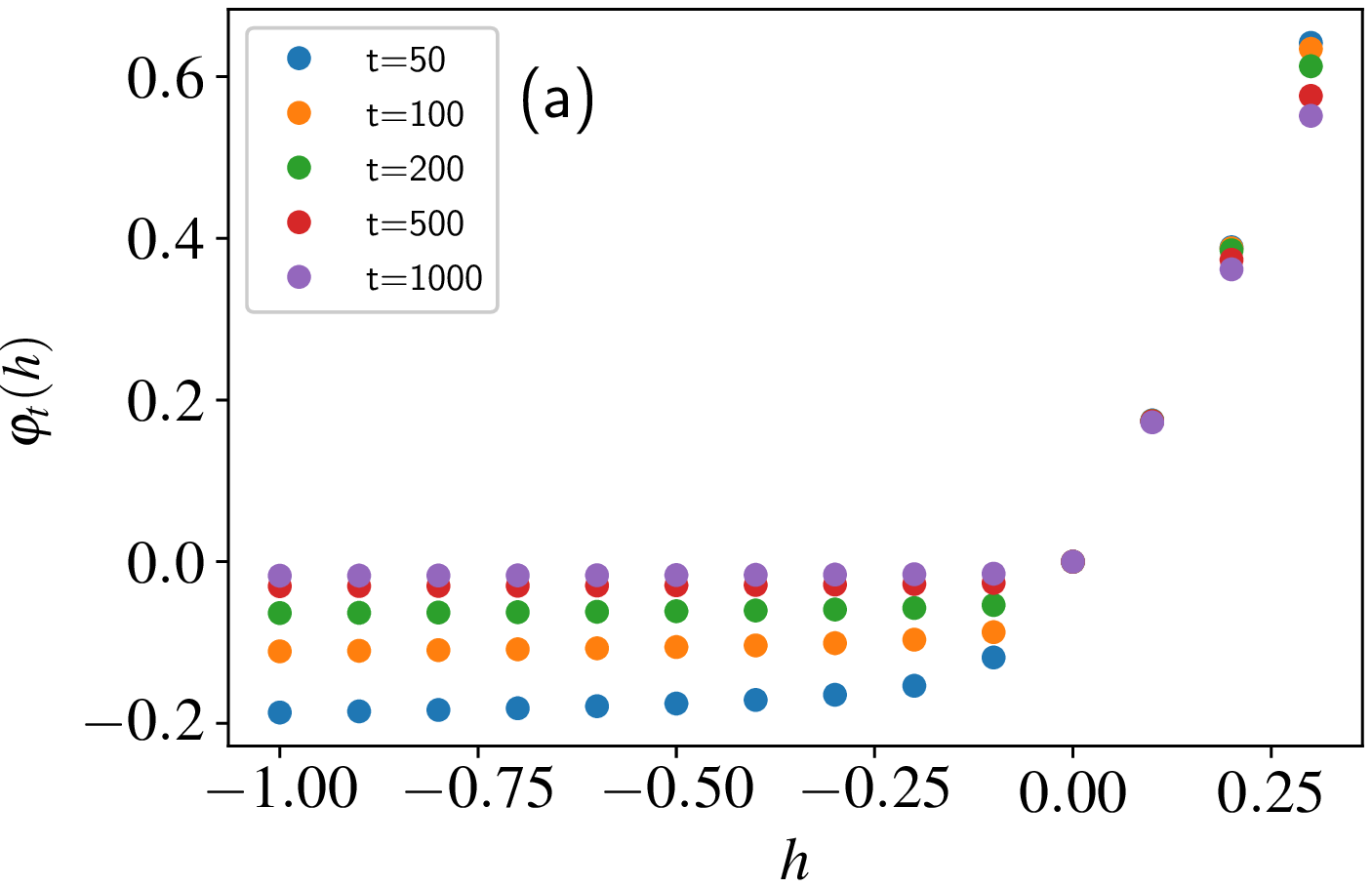}  \hspace{0.5cm}
		\includegraphics[clip,width=7.5cm]{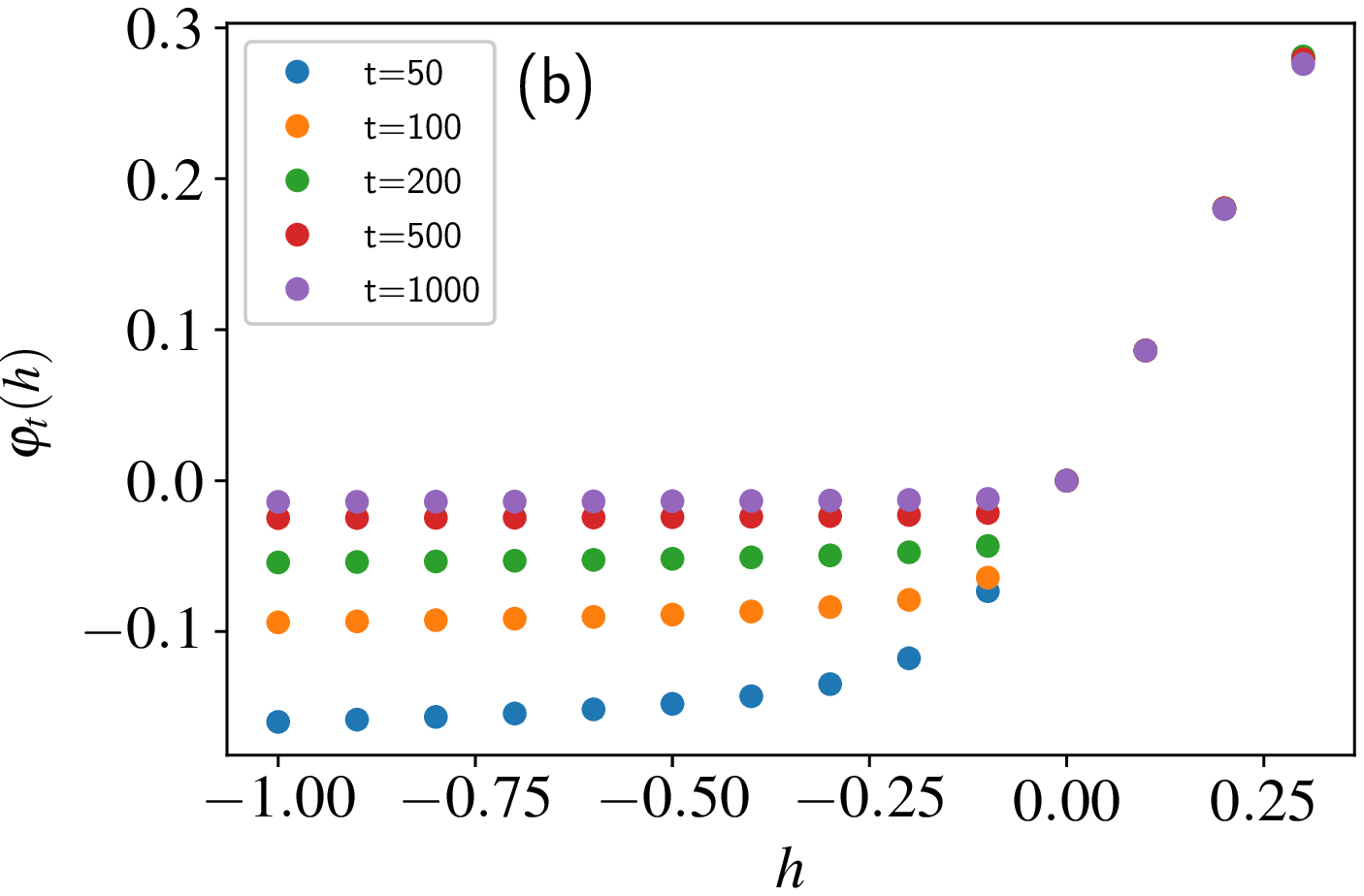}  \hspace{0.5cm}
		\includegraphics[clip,width=7.5cm]{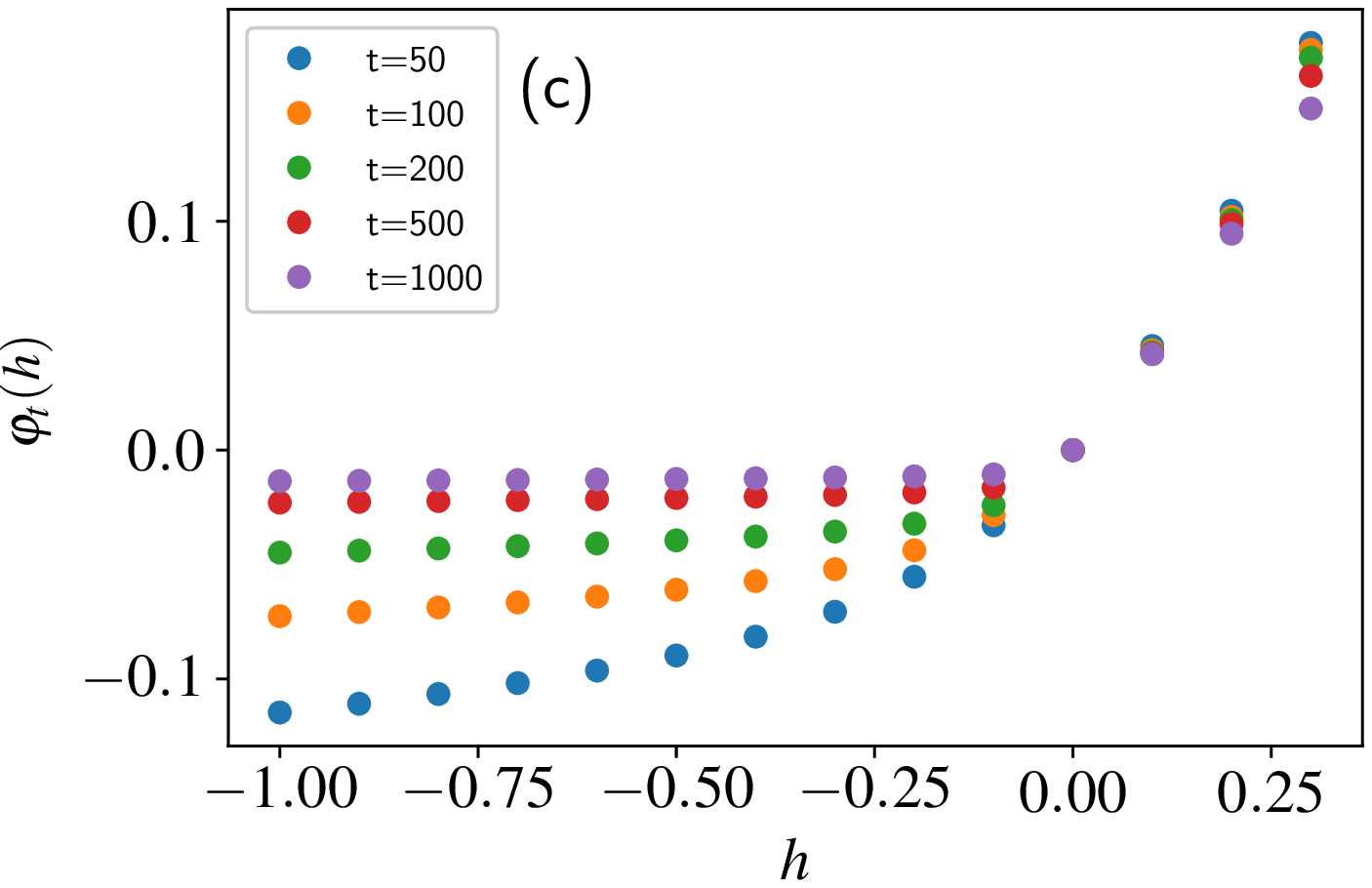}  \hspace{0.5cm}
		\caption{  Numerical results of the finite-time CGF ($(1/t) \log \mathbb E [e^{h N_t}] $) of the counting process with Pareto distribution with $m=3.5$ (a), inverse Rayleigh distribution with $\beta = 1$ (b) and log-normal waiting-time distribution with $\mu=0$ and $\sigma=1.5$ (c). We observe the emergence of the affine part when $h$ is negative. 
		}
		\label{fig:e}
		\vspace{-0.5cm}
	\end{center}
\end{figure}

We perform numerical simulations with these waiting time distributions and  plot in Fig.\ref{fig:e} the finite-time cumulant generating function (CGF) of the counting process $N_t$, defined as 
\begin{equation}
\varphi_t(h) : = \frac{1}{t} \log M(t,h).
\end{equation}
 The figure shows the emergence of the affine part for the negative $h$, which means that $M(t,h)$ decreases sub-exponentially. 
 Theorem \ref{thm:normconv} states that, when the waiting time distribution is the Pareto distribution with an integer exponent $m$, this sub-exponential decrease is proportional with $M_0(t)$ (where $M_0(t)=1-F(t)$) with a coefficient $1/(1-e^h)$.  To study if the same statement  is satisfied with the various waiting-time distributions introduced above, we then plot $M(t,h)/M_0(t)$ (where $M_0(t)=1-F(t)$) in Fig.\ref{fig:theorem2} as a function of $t$. 
In the figure, $M(t,h)/M_0(t)$ seems to converge to $1/(1-e^h)$ for Pareto distribution ($m=3.5$) and inverse Rayleigh distribution, while it converges to a value close to $1/(1-e^h)$ for lognormal distribution. We note that the Pareto and inverse Rayleigh distribution are both regularly varying at infinity while the log-normal is not. It is an interesting future perspective to quantitatively understand these convergences and prove the corresponding Theorem \ref{thm:normconv_0} for these waiting time distributions.

\begin{figure}
	\begin{center}
		\includegraphics[clip,width=7.cm]{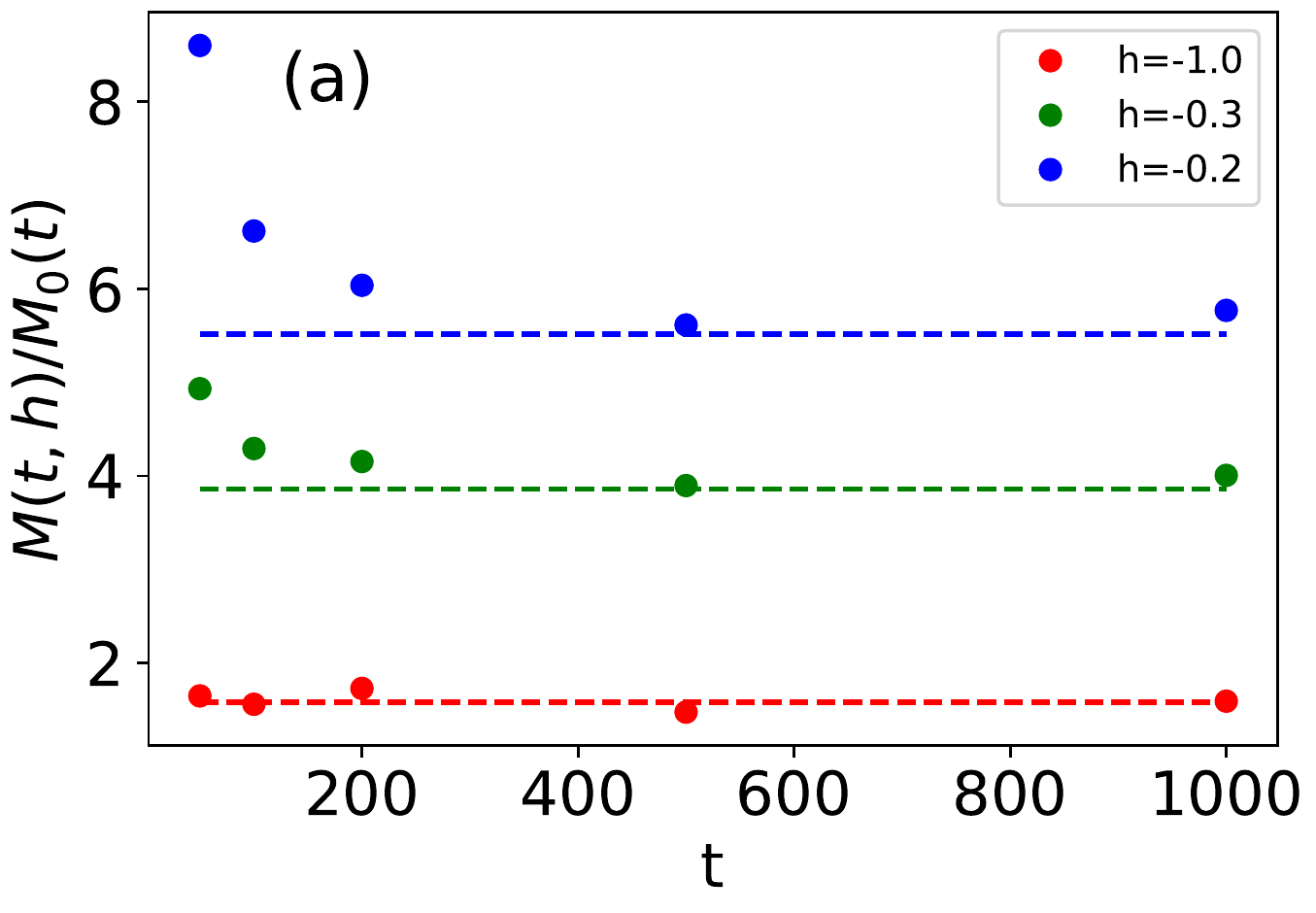}  \hspace{0.5cm}
		\includegraphics[clip,width=7.5cm]{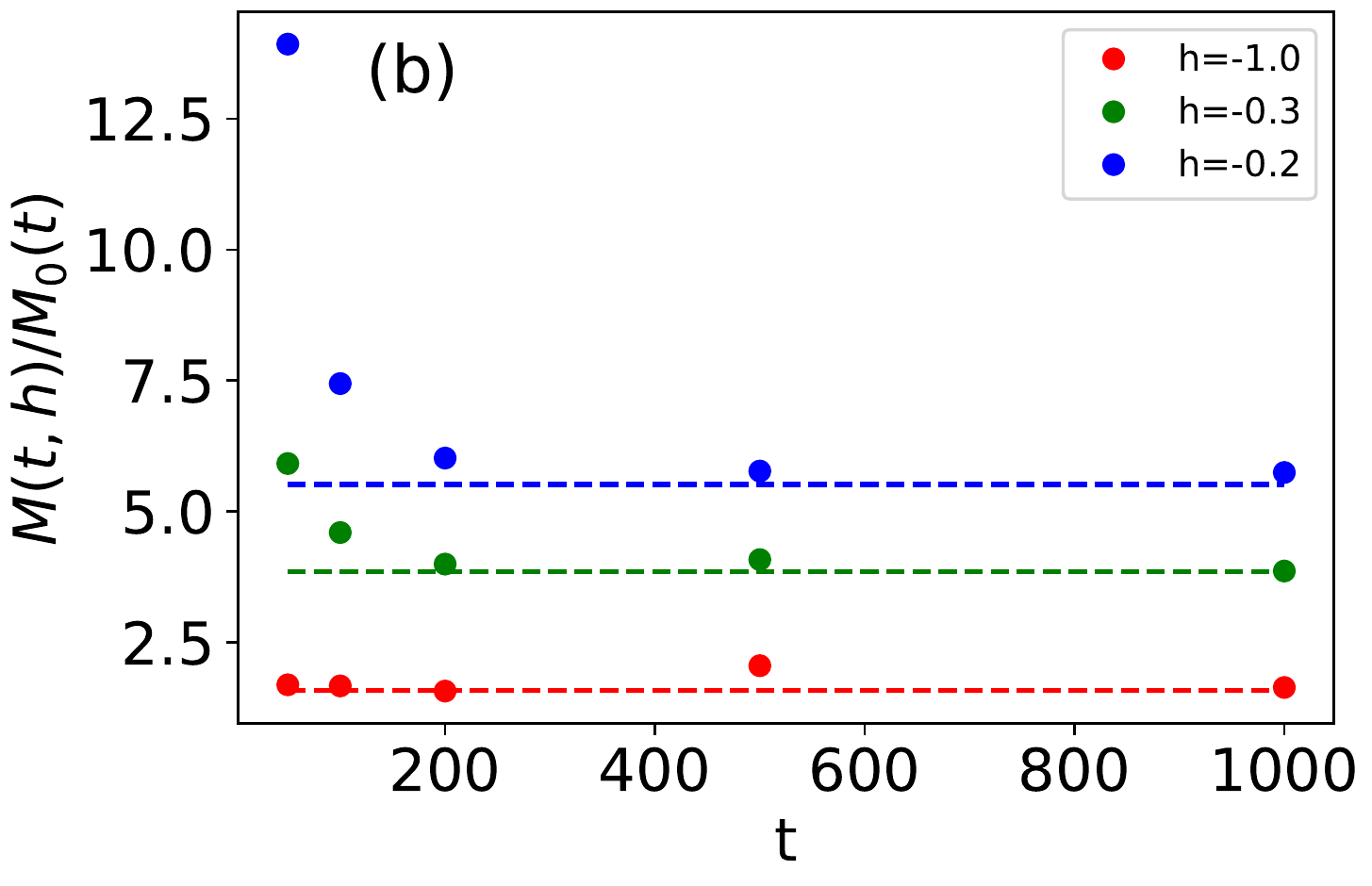}  \hspace{0.5cm}
		\includegraphics[clip,width=7.5cm]{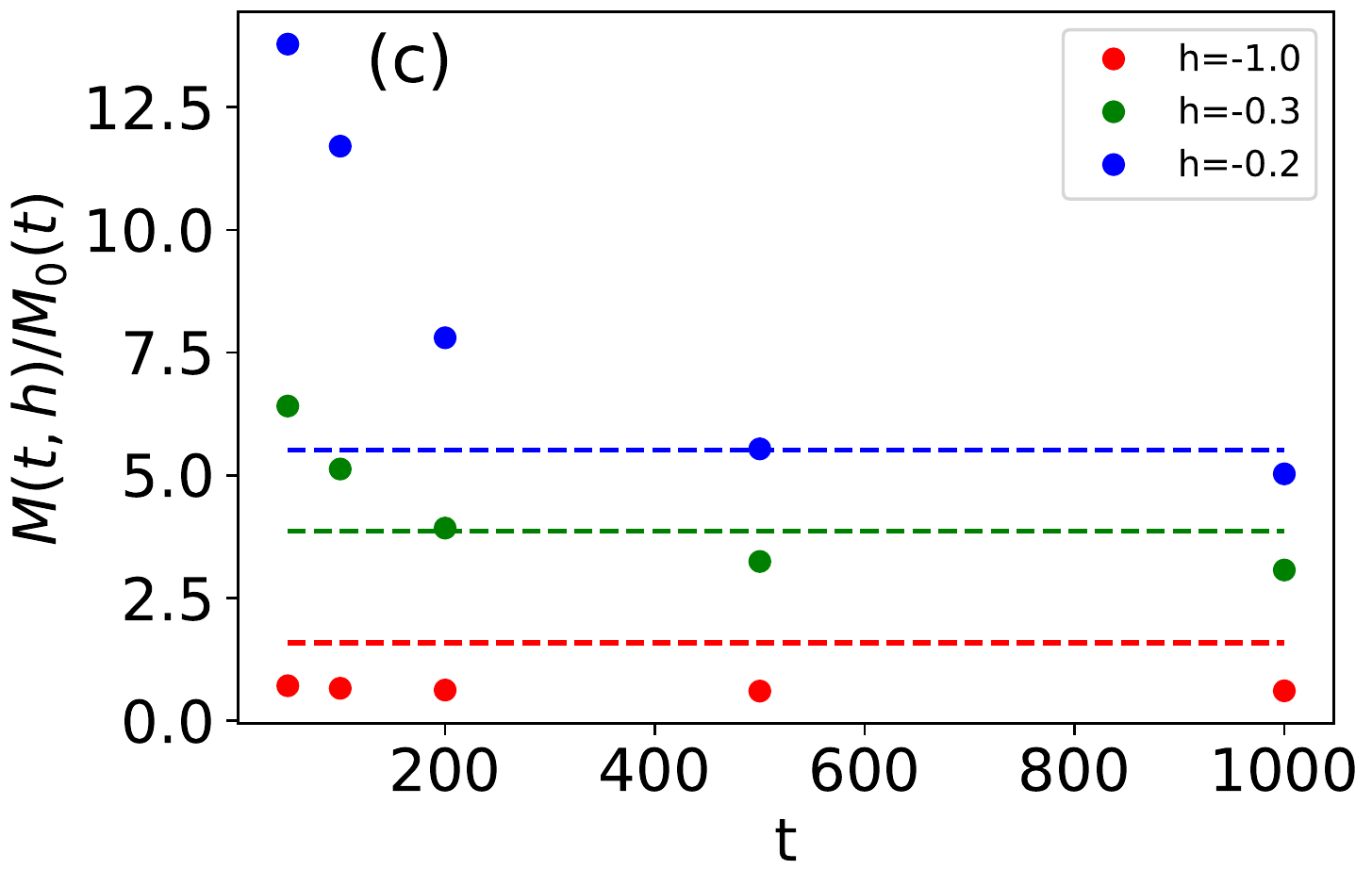}  \hspace{0.5cm}
		\caption{ $M(t,h)/M_0(t)$ (filled circles) together with $1/(1-e^h)$ (dashed lines) for Pareto distribution with $m=3.5$ (a), inverse Rayleigh distribution with $\beta = 1$ (b) and log-normal waiting-time distribution with $\mu=0$ and $\sigma=1.5$ (c). 
		}
		\label{fig:theorem2}
		\vspace{-0.5cm}
	\end{center}
\end{figure}


\subsection{Study of the rate function}
\label{sec:numerical_ratefunction}

In this section, we numerically study the asymptotic behaviour of the finite-time rate function $i_t(x)$ 
\begin{equation}
i_t(x) := -\frac{1}{t} \log \mathbb P \left (\frac {N_t}{t}\simeq x \right ).
\end{equation}
We first plot $i_t(x)-\min_x i_t(x)$ for several $t$ with different waiting-time distributions (introduced in the previous sub-section) in Fig.~\ref{fig:ratefunction}. The figure indicates the emergence of the affine part, which manifests anomalously large probability of rare fluctuations where $N_t/t$ takes smaller values than the expectation. To study the finite-time asymptotic, we then plot $\log \mathbb P[N_t<xt]$ (for a fixed $x$) as a function of $t$ in Fig.~\ref{fig:ratefunction_2}. We observe the asymptotic behaviour of $\log \mathbb P[N_t<xt]$ as
\begin{equation}
\log \mathbb P[N_t<xt]  \sim a \log M_0(t) + \log(t) + b
\label{eq:numerical_observation}
\end{equation}
with constants $a,b$ (that can potentially depend on $x$). For Pareto and the inverse Rayleigh waiting time distributions, $a$ seems to be 1, while $a$ is different from 1 for the log-normal waiting-time distribution.

\begin{figure}[htbp]
	\begin{center}
		\includegraphics[clip,width=7.5cm]{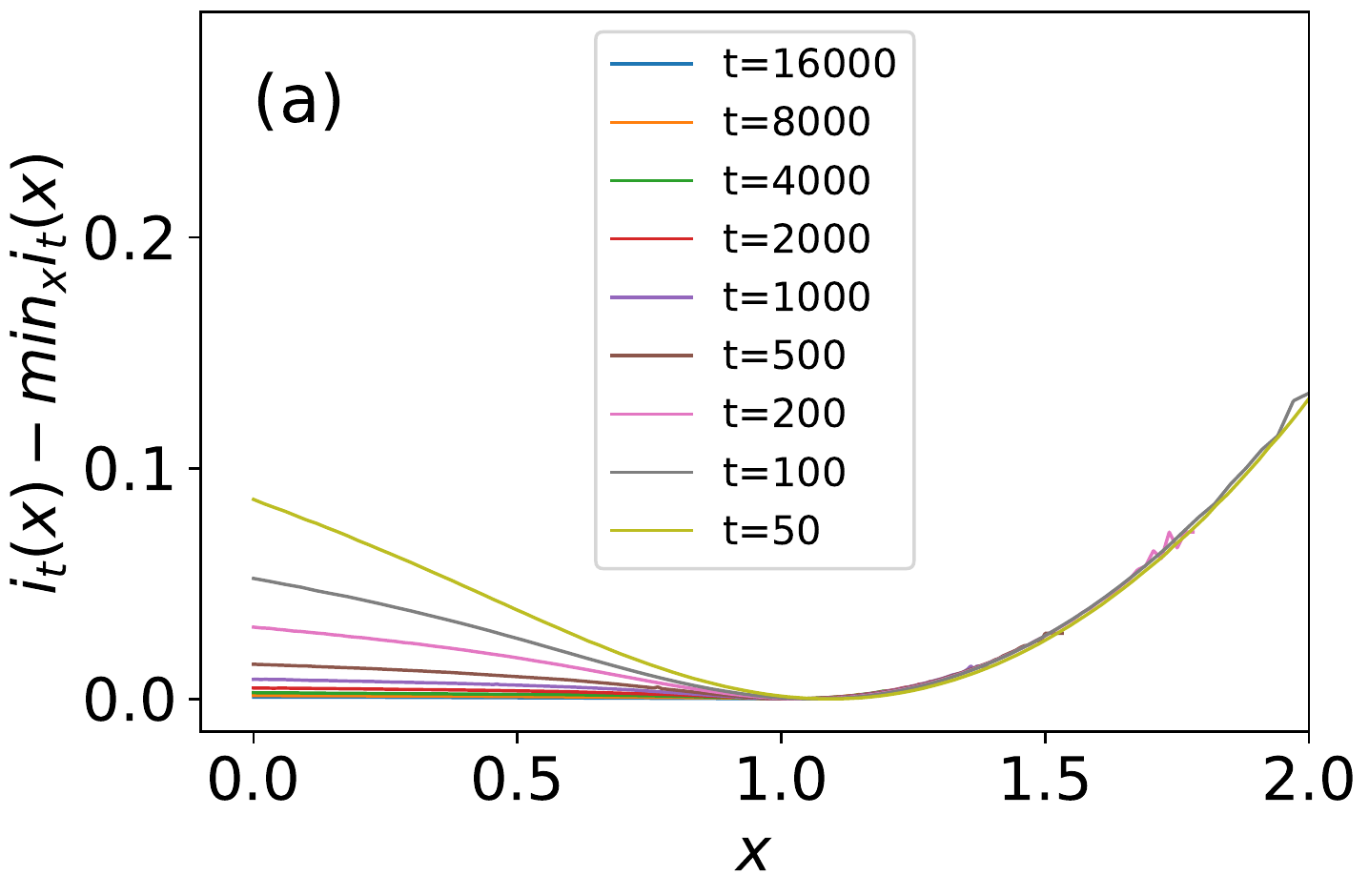}  \hspace{0.5cm}
		\includegraphics[clip,width=7.5cm]{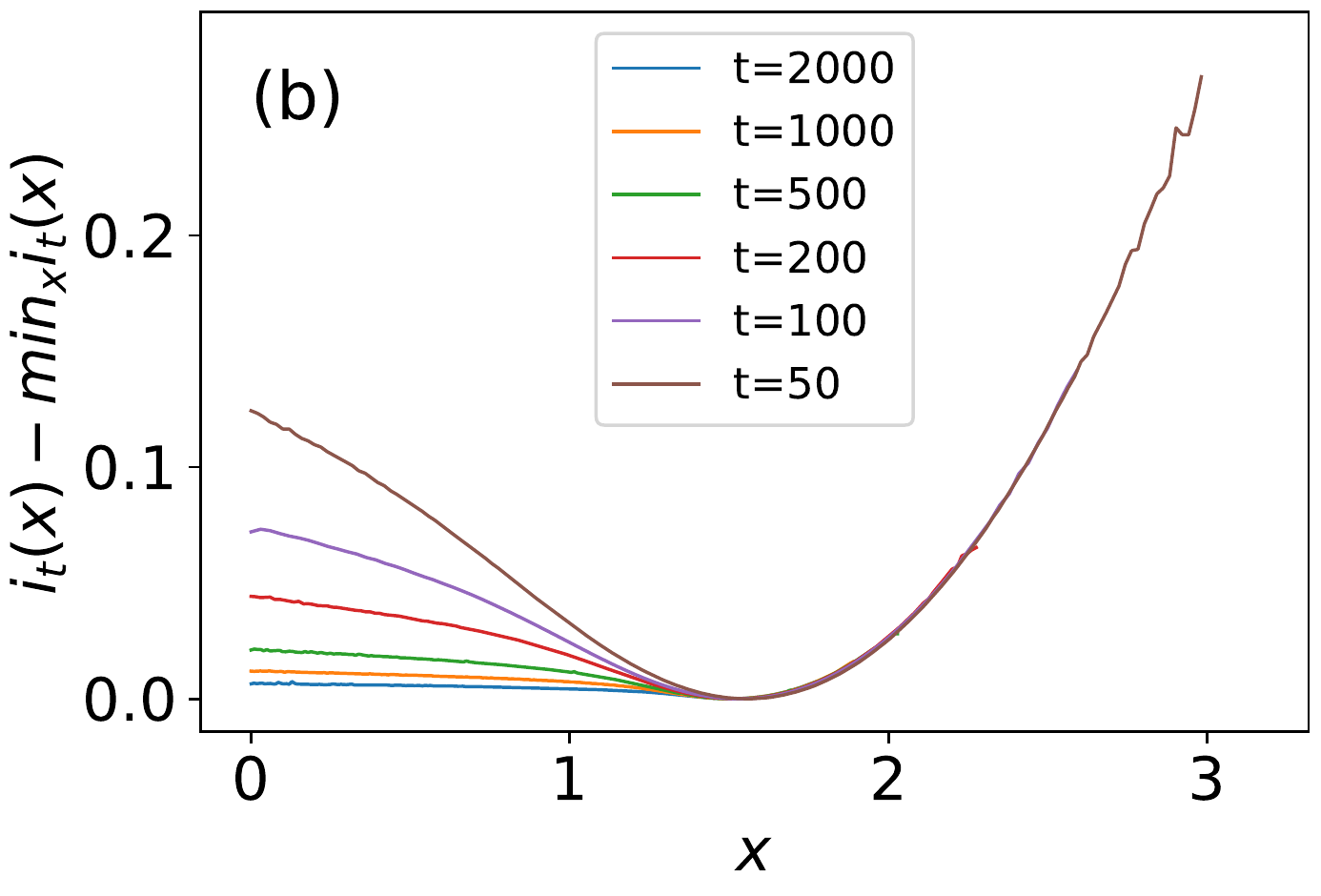}  \hspace{0.5cm}
		\includegraphics[clip,width=7.5cm]{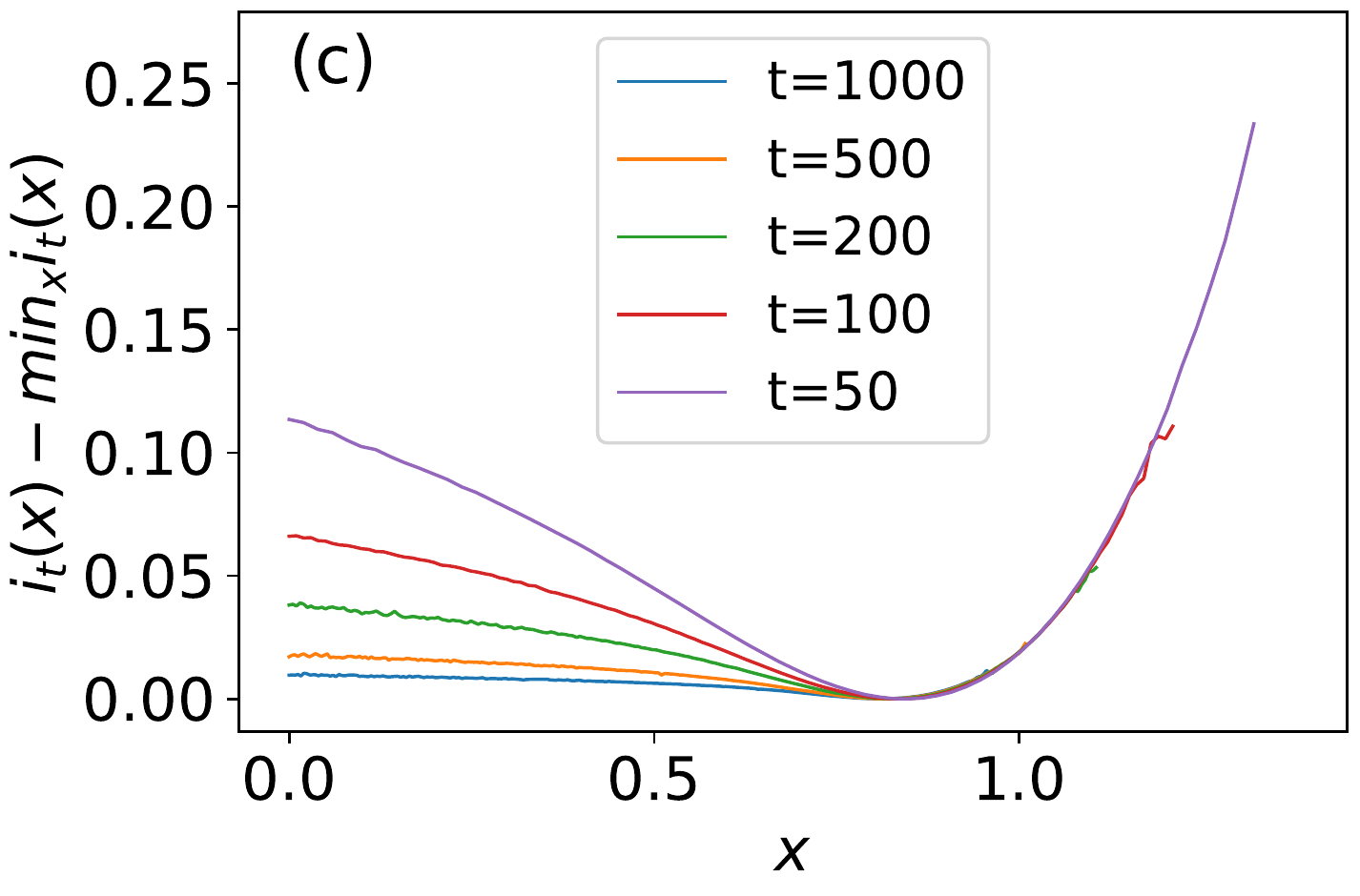}  \hspace{0.5cm}
		\includegraphics[clip,width=7.5cm]{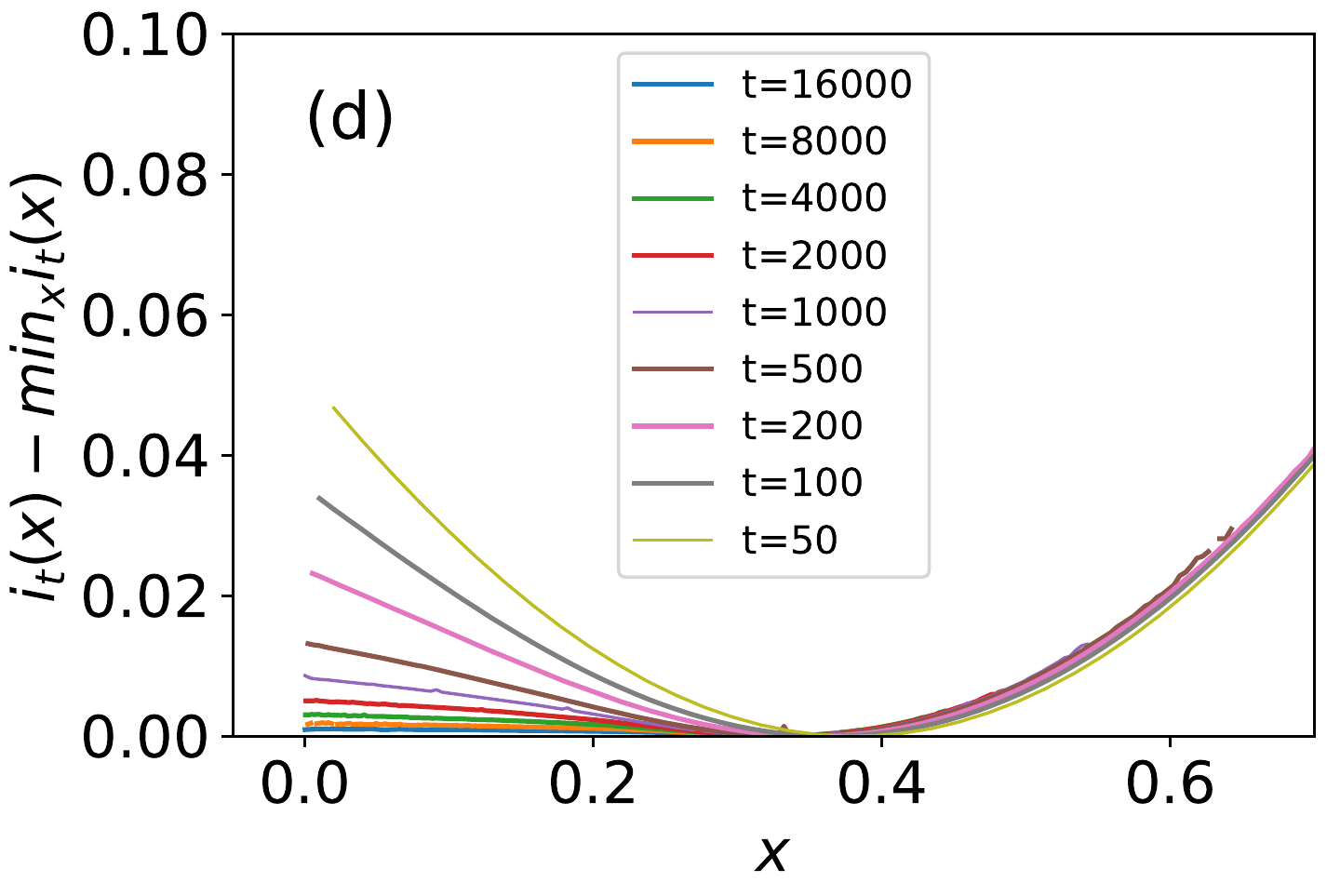}  \hspace{0.5cm}
		\caption{ The finite time rate function $i_t(x)-\min_x i_t(x)$ for several $t$ with different waiting time distributions: (a) Pareto distribution with $m=3$, (b) Pareto distribution with $m=3.5$, (c), inverse Rayleigh distribution with $\beta = 1$ and (d) the log-normal waiting-time distribution with $\mu=0$ and $\sigma=1.5$. In all  panels,  the affine part emerges as $t$ increases. 
		}
		\label{fig:ratefunction}
		\vspace{-0.5cm}
	\end{center}
\end{figure}

Mathematically justifying this asymptotic expression is an open problem. For example, we can immediately derive a lower bound for $\mathbb P[N_t < xt]$ for any $x> 0$ as
	\begin{equation}
		\mathbb P[N_t < xt]\geq \mathbb P[N_t=0]= \mathbb P[\tau_1>t]=M_0(t),
	\end{equation}
which is consistent with the observation. For the upper bound it is natural to use the relation
\begin{equation}
	\mathbb P[N_t<xt]=\mathbb P[S_{\lfloor xt \rfloor}>t]
\end{equation}
and the bound (\ref{Snbound}) with replacing $n$ by $\lfloor xt \rfloor$. This leads to
\begin{equation}
	\mathbb P[S_n\geq t]
	\leq n M_0(t)+n^2 C_m\frac{d^m}{(d+t)^m}\left[1+\frac{\bar c}{d}\right]^{n-1}
\end{equation}
valid for any $d\geq 1$. Unfortunately the bound provided by this inequality (and in particular its second term) is not strong enough to derive a meaningful asymptotic expression.  Indeed, it is easy to see that when $t\to\infty$ the second term diverges even if one takes $d\to\infty$ first. 
Interestingly, the logarithm of the first term yields
\begin{equation}
\log(\lfloor xt\rfloor M_0(t))\sim \log M_0(t) + \log(x t) 
\end{equation}
as $t\to\infty$, which is coincidentally the same as the observed expression (\ref{eq:numerical_observation}) for Pareto and inverse Rayleigh waiting time distributions. This implies that refining the bound (\ref{Snbound}) could be the key to derive (\ref{eq:numerical_observation}) at least for these waiting time distributions. Pursuing this direction is out of scope in the current manuscript but an interesting future perspective.  

\begin{figure}[htbp]
	\begin{center}
		\includegraphics[clip,width=7.5cm]{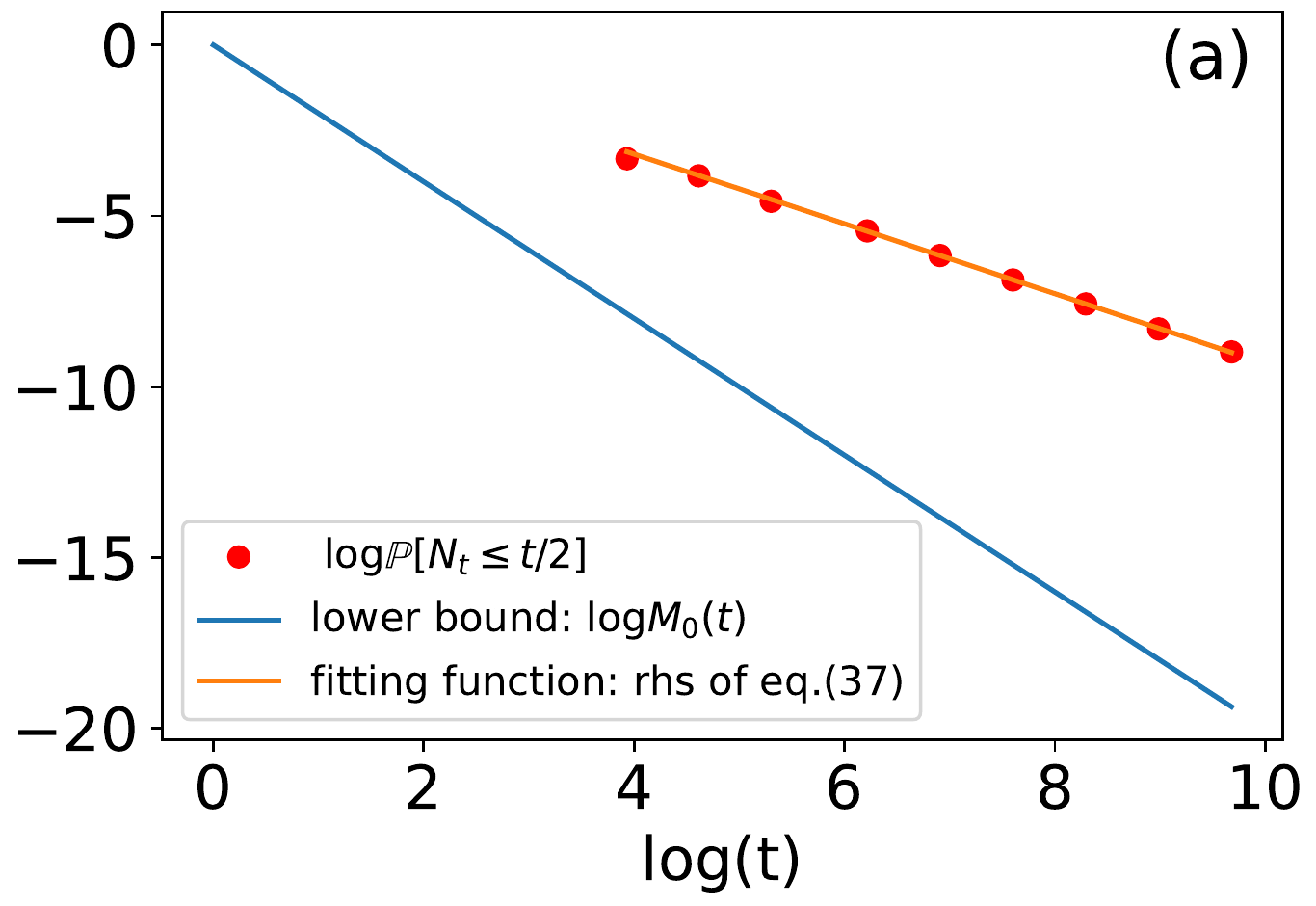}  \hspace{0.5cm}
		\includegraphics[clip,width=7.5cm]{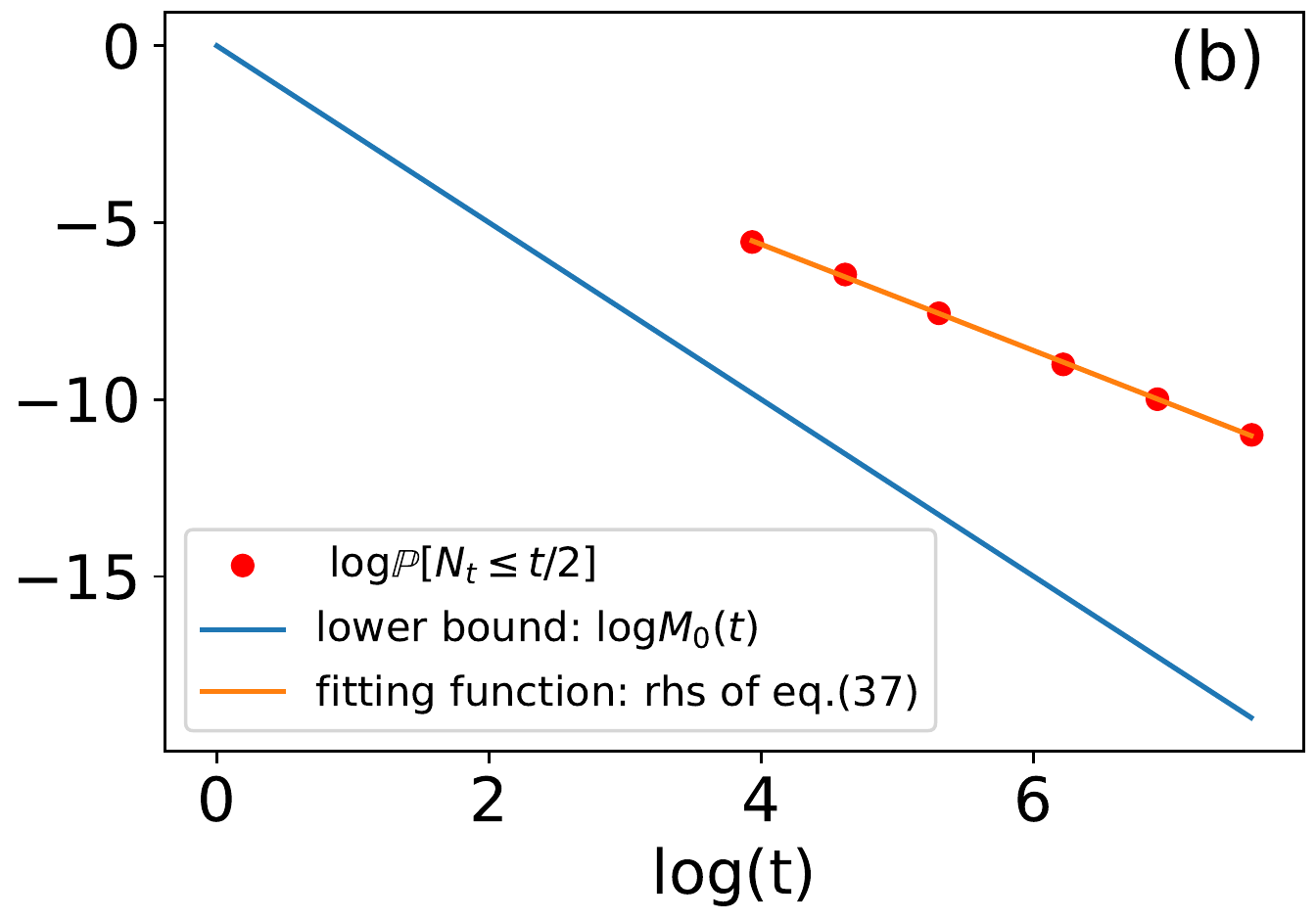}  \hspace{0.5cm}
		\includegraphics[clip,width=7.5cm]{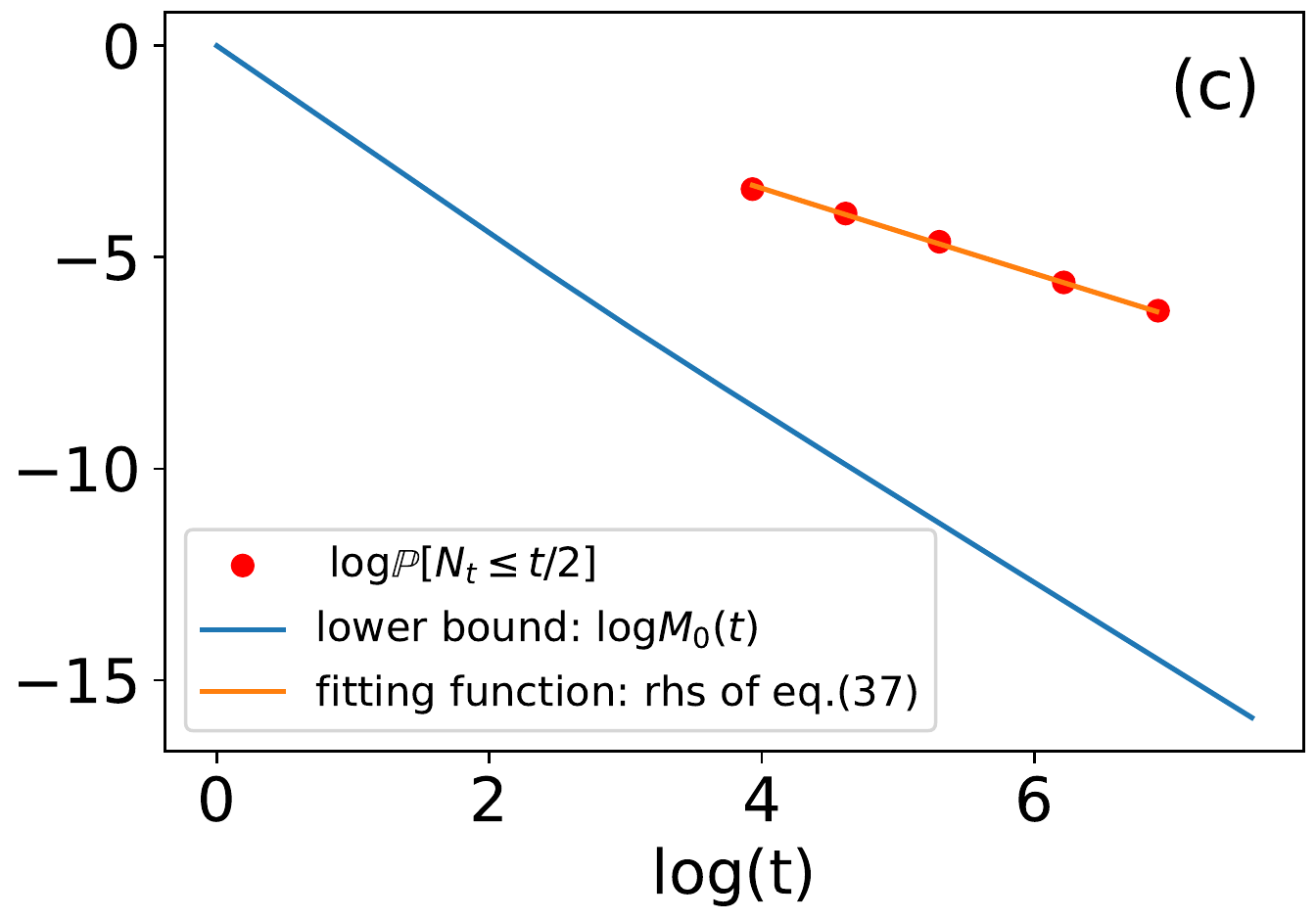}  \hspace{0.5cm}
		\includegraphics[clip,width=7.5cm]{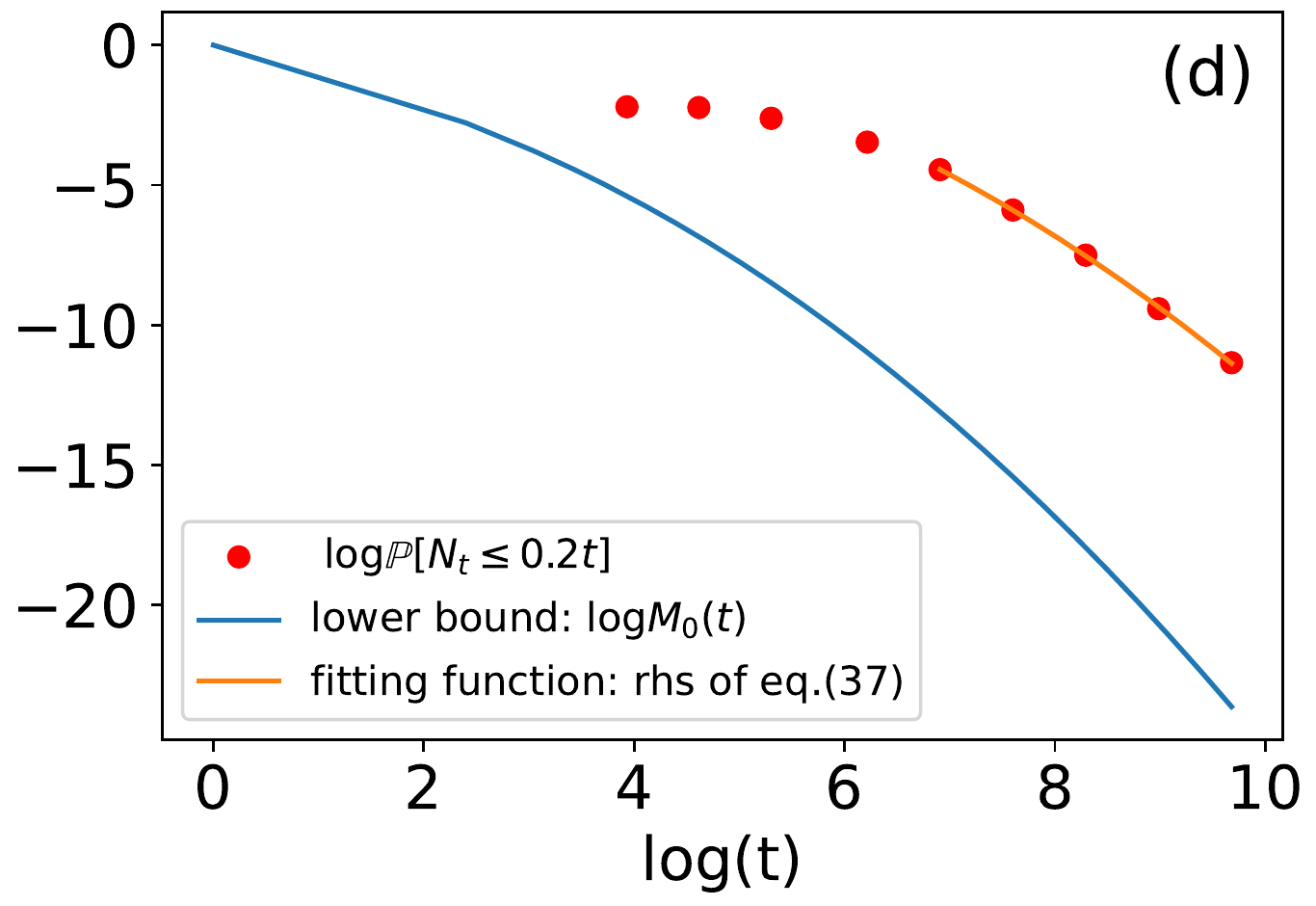}  \hspace{0.5cm}
		\caption{ $\log \mathbb P[N_t<xt]$ for several waiting time distributions: (a) Pareto distribution with $m=3$, (b) Pareto distribution with $m=3.5$, (c), inverse Rayleigh distribution with $\beta = 1$ and (d) the log-normal waiting-time distribution with $\mu=0$ and $\sigma=1.5$. The lower bound $\log M_0(t)$ and a fitting function $a \log M_0(t) + \log (t) + b$ (with fitting parameters $a,b$) are also plotted. These fitting parameters are determined as $a=1.01, b=0.92$ for (a), $a=1.00, b=0.43$ for (b), $a=1.00, b=2.00$ for (c) and $a=0.89, b=2.51$ for (d). 
		}
		\label{fig:ratefunction_2}
		\vspace{-0.5cm}
	\end{center}
\end{figure}

\section{Conclusion}
\label{seq:conclusion}

In this article, we studied the finite-time asymptotic of the MGF in a counting process $N_t$ with  Pareto distributions \eqref{eq:pareto}. To this end, we applied the result (\ref{feller1}) and we proved an explicit expression of the bounds for the finite-time MGF (Theorem \ref{thm:normconv}) using an expansion approach. The method to prove these expressions could be applied to more general cases, as the validity of the relation (\ref{feller1}) is not restricted to Pareto distribution. Also we expect that the bound of Theorem \ref{thm:normconv} may be extended beyond the case where one can perform a partial fraction decomposition to estimate the integrals.  The same affine part has been  observed in heat currents with the inverse Rayleigh distribution \cite{b14} and in the counting process with  lognormal distribution (Fig.\ref{fig:e}). Similar finite-time asymptotics for the MGF are  anticipated  as numerically demonstrated in Section~\ref{seq:appendixa}. Studying how the methods of this article can be generalised in these cases would be an interesting future perspective.

In physics, there have been tremendous efforts to understand and characterise singularities appearing in large deviation functions (rate functions and CGFs). In equilibrium statistical physics,  these singularities correspond to phase transitions because the large deviation functions are the thermodynamic functions (\cite{callen1998thermodynamics, ellis2006entropy,  b4, jona2015large, bertini2015macroscopic} for instance). Studying large deviation functions of time-averaged quantities in non-equilibrium systems have attracted a strong attention of statistical physicists since the discovery of fluctuation theorem in 1993 \cite{evans1993probability,kurchan1998fluctuation, lebowitz1999gallavotti,b4,b7}. The singularities of the large deviation functions are related to dynamical phase transitions and are studied in lattice gas models \cite{bodineau2005distribution, appert2008universal, bodineau2008long, bodineau2012finite, baek2017dynamical, shpielberg2017geometrical, shpielberg2018universality},  high-dimensional chaotic dynamics \cite{tailleur2007probing, laffargue2013large, bouchet2014stochastic}, glass formers \cite{hedges2009dynamic, garrahan2009first, jack2010large, pitard2011dynamic, limmer2014theory, nemoto2017finite, speck2012first}, diffusive hydrodynamic equations \cite{bertini2005current, hurtado2011spontaneous, tizon2017structure} and active matters \cite{vaikuntanathan2014dynamic, cagnetta2017large, whitelam2018phase, nemoto2019optimizing}. In these studies, the dynamics are usually defined as a Markov process (or deterministic process), so that the singularity does not appear whenever finite size systems are considered (see \cite{jack2020dynamical} for an illustrative example). In our case, the singularity in the large deviation functions is already present for a finite-size system because of the long-range correlation in time. Studying how this difference  can alter the well-studied dynamical phase transitions in physics context would be an interesting perspective, in the similar way that heat conductions in aerogels were  studied using renewal-reward processes \cite{b14}.

\section*{Acknowledgement}
H.H. is in the Cofund MathInParis PhD project that has received funding from the European Union’s Horizon 2020 research and innovation programme under the Marie Sk lodowska-Curie grant agreement No.~754362.~\includegraphics[width=0.7cm]{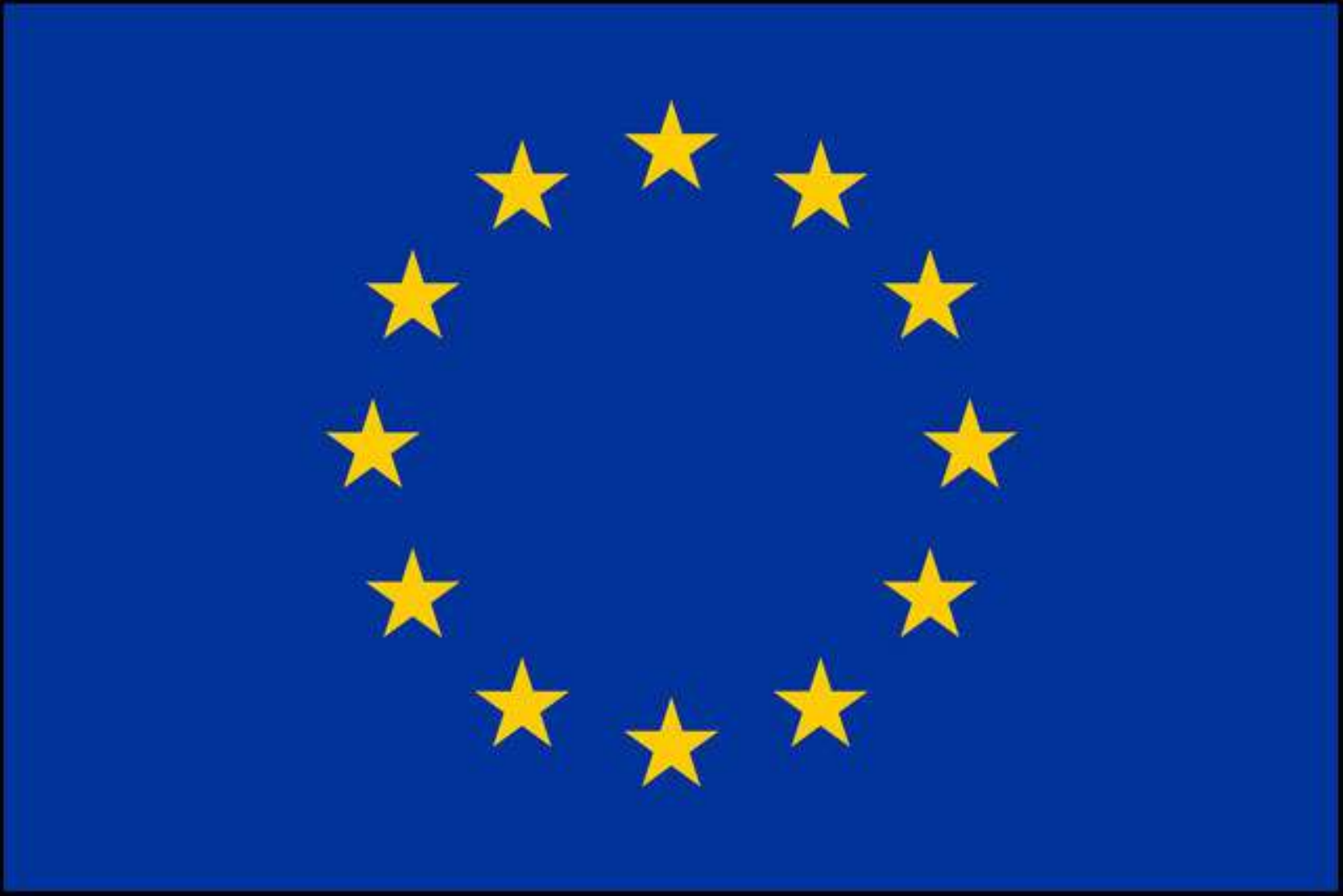}
R.L. is supported by the ANR-15-CE40-0020-01 grant LSD. The authors thank  Robert Jack  for fruitful discussions about finite-time CGFs and rate functions. H.H. also acknowledges the support from Robert Jack during his stay in the University of Cambridge for his mobility project.

\bibliographystyle{plain_url}
\bibliography{draft.bib}

\end{document}